\documentclass[12pt, a4paper, reqno]{amsart}
\usepackage{mathrsfs}
\usepackage{bbm}
\usepackage{pgf,tikz}
\usepackage{tabu}
\usepackage{amsmath,amscd,amssymb,latexsym}
\usepackage[all]{xy}

\input xypic

\evensidemargin=\oddsidemargin
\DeclareMathVersion{can}
\DeclareMathAlphabet{\can}{OT1}{cmss}{m}{n}
\newtheorem{thm}{Theorem}[section]
\newtheorem{cor}[thm]{Corollary}
\newtheorem{lem}[thm]{Lemma}
\newtheorem{prop}[thm]{Proposition}
\newtheorem{rem}[thm]{Remark}
\newtheorem{exa}[thm]{Example}
\theoremstyle{definition}
\newtheorem{defn}[thm]{Definition}
\theoremstyle{fact}

\theoremstyle{conjecture}

\numberwithin{equation}{section}

\newcommand{\Tr}{\operatorname{Tr}}

\begin{document}
\title[Characterization of $p$-ary functions in terms of association schemes]
{Characterization of  $p$-ary functions in terms of association schemes and its applications}

\subjclass[2000]{11T71, 97K30, 94B05, 94C10}
\keywords{ $p$-ary function; association scheme; Walsh transform; few-weight linear code}

\author[Y. Wu]{ Yansheng Wu}
\address{\rm  School of Computer Science, Nanjing University of Posts and Telecommunications, Nanjing 210023, China; Shanghai
Key Laboratory of Trustworthy Computing, East China Normal University,
Shanghai, 200062, China}
\email{yanshengwu@njupt.edu.cn}

\author[J. Y. Hyun]{Jong Yoon Hyun}
 \address{\rm Konkuk University, Glocal Campus, 268 Chungwon-daero Chungju-si Chungcheongbuk-do 27478, South Korea}
\email{hyun33@kku.ac.kr}

\author[Y. Lee]{ Yoonjin Lee}
\address{\rm Department of Mathematics, Ewha Womans
University, Seoul 120-750, South Korea}
\email{yoonjinl@ewha.ac.kr}




\date{\today}

\thanks{\tiny
Y. Wu was supported by the National Natural Science Foundation of China (Grant No. 12101326), the Natural Science Foundation of Jiangsu Province (Grant No. BK20210575 ). J.Y. Hyun was supported by the National Research Foundation of Korea(NRF) grant
funded by the Korea government(MEST)(NRF-2017R1D1A1B05030707).
Y. Lee was supported by Basic Science Research Program
through the National Research Foundation of Korea(NRF) funded by the Ministry of Education
(Grant No. 2019R1A6A1A11051177) and also by the National Research Foundation of Korea(NRF)
grant funded by the Korea government (MEST)(NRF-2017R1A2B2004574).}

\begin{abstract}
We obtain an explicit criterion for $p$-ary functions to produce association schemes in terms of their Walsh spectrum. Employing this characterization, we explicitly find a correlation between $p$-ary bent functions
and association schemes; to be more exact, we prove that a $p$-ary bent function induces a $p$-class association scheme if and only if the function is weakly regular. As applications of our main criterion, we construct many infinite families of few-class association schemes arising from $p$-ary functions. Furthermore, we present four classes of $p$-ary two-weight linear codes, which are constructed from the association schemes produced in this paper.
\end{abstract}

\maketitle
\section{Introduction}

Association schemes are originated from combinatorial design theory and character theory of finite groups.
In algebraic combinatorics \cite{BI, G}, association schemes are connected with both combinatorial designs and coding theory~\cite{DL}. There have been many developments on construction of association schemes in an algebraic approach \cite{D} as well as using many objects such as Schur rings \cite{M}, nonlinear functions \cite{BP, OP,PTFL,V},  partitions \cite{GL,ZE1, ZE2}, and so on.

Recently, there have been active developments on construction of association schemes using bent functions or plateaued functions. Introducing some of the previous results,  Tan  {\em et al.}  \cite{TPF} constructed many strongly regular graphs by using ternary bent functions and then Pott {\em et al.} \cite{PTFL} obtained $p$-class association schemes by employing $p$-ary weakly regular bent functions, where $p$ is an odd prime number. Afterwards, Mesnager and Sinak \cite{MS} extended the construction of association schemes to weakly regular plateaued functions. Most recently, \"Ozbudak and  Pelen \cite{OP} constructed some $2$-class symmetric association schemes from non-weakly regular bent functions.
Inspired by the previous works, it is natural to investigate a characterization of $p$-ary functions in terms of
association schemes.

In this paper, we obtain an explicit criterion for $p$-ary functions to produce association schemes in terms of their Walsh spectrum. Employing this characterization, we explicitly find a correlation between $p$-ary bent functions
and association schemes; to be more exact, we prove that a $p$-ary bent function induces a $p$-class association scheme if and only if the function is weakly regular. In fact, there is a previous result on the study of connection of $p$-ary bent functions with strongly regular graphs \cite[Theorem A]{HL}.  As applications of our main criterion, we construct many infinite families of $t$-class association schemes arising from $p$-ary functions, where $t=2, 3, 4, 5, 6.$ Furthermore, we present four classes of $p$-ary two-weight linear codes, which are constructed from the association schemes produced in this paper.

Our paper is organized as follows. In Section 2, we introduce some necessary definitions and some previous results. In Section 3, using the Walsh spectrum of $p$-ary functions, we obtain a main result on association schemes (Theorem~\ref{eqn:thm3.4}). In Section 4, we consider some applications of the main result. Specifically, we first prove
that a $p$-ary bent function induces a $p$-class association scheme if and only if the function is weakly regular (Theorem 4.2). Secondly, we construct several classes of association schemes by considering the $p$-ary functions such that $f(ax)=f(x)$ for all $(a,x)\in\mathbb{F}^*_p\times\mathbb{F}_q$. We find some 2-class, 3-class, 4-class, 5-class, and 6-class association schemes (Propositions 4.6, 4.9, 4.11). Finally, we get some two-weight linear codes from the schemes (Proposition 4.15).
\section{Preliminaries}











\subsection{Symmetric association schemes}$~$


In this section, we introduce some basic notions regarding association schemes, weakly regular $p$-ary bent functions, and Fourier-reflexive partitions.

\begin{defn}Let $X$ be a finite set. A \emph{symmetric $d$-class association scheme}
$(X, \{R_l\})_{i=0}^d$
is a partition of $X\times X$ into binary relations  $R_0,R_1,\ldots,R_d$ satisfying the following properties:\\
$\bullet$ $R_0=\{(x,x):x\in X\};$\\
$\bullet$ $R_i$ is symmetric for $i=1,2,\ldots,d,$ that is, $(x,y)\in R_i$ if and only if $(y,x)\in R_i;$\\
$\bullet$ for all $i,j,k \in \{0,1,\ldots,d\}$ there is an integer
$p^k_{ij}$ such that for all $(u,v)\in R_k$,
\[
p^k_{ij}=|\{w\in X:(u,w)\in R_i \mbox{ and } (w,v)\in R_j\}|.
\]
\end{defn}

One of the well-known constructions  of association schemes uses the Schur rings as follows.
Let $G=\cup_{i=0}^dD_i$ be a finite abelian group, where the union is pairwise disjoint and
$D_i\neq\emptyset$ for every $i=1, 2, \ldots, d$, and the following properties are satisfied:

$\bullet$ $D_0 = \{1_G\}$;

$\bullet$ $D^{-1}_i= D_i$ for any $i\in\{0, 1,\ldots, d\}$, where $D^{-1}= \{g^{-1} : g \in D_i\}$;

$\bullet$ $D_iD_j =\sum_{k=0}^d p^k_{ij}D_k$ for any $i$ and $j$ with $0 \leq i, j\leq d$, where $p^k_{ij}$ are integers.

Then the subset $\langle D_0,\dots,D_d\rangle$ in $\mathbb{C}[G]$ spanned by $D_0,\ldots,D_d$ is
so-called a {\it Schur ring} over $G$. The configuration $(G, {R_i})_{0\leq i\leq d}$ forms a
$d$-class symmetric association scheme on $G$, where $R_i = \{(g, h) : gh^{-1}\in D_i\}$ for $0 \leq i \leq d$.

\subsection{Weakly regular $p$-ary bent functions}$~$

Let $f$ be a $p$-\emph{ary function} from $\mathbb{F}_{q}$ to $\mathbb{F}_p$, where $q=p^m$.
We define $D_{f,i}$ associated with $f$ by the set
\[
D_{f,i}=\{\beta\in\mathbb{F}_q:f(\beta)=i\}.
\]
Let $\zeta_p$ be a primitive $p$-th root of unity in a complex field $\mathbb C$.
The \emph{Walsh-Hadamard transform} $W_f$ of a $p$-ary function $f$ is a complex-valued
function of $\mathbb{F}_{q}$ defined by
\[
W_f(\beta)=\sum\limits_{x\in \mathbb{F}_{q}}\zeta^{f(x)-\operatorname{Tr}(\beta x)}_p,
\]
where $\Tr$ is the trace function from $\mathbb{F}_q$ to $\mathbb{F}_p$.
The inverse Walsh-Hadamard transform of a $p$-ary function $f$ is given by
\[
\zeta^{f(\beta)}_p=p^{-m}\sum\limits_{x\in \mathbb{F}_{p^m}}W_f(x)\zeta^{\Tr(\beta x)}_p.
\]

A $p$-ary function $f$ is called
$\emph{bent}$ if $|W_f(\beta)|^2=q$ for any $\beta\in\mathbb{F}_{q}$.
It is known that if $f$ is a $p$-ary bent function, then
\begin{align}\label{eqn:(1)}
W_f(\beta)=\left\{
\begin{array}{c l}
\pm p^{\frac{m}{2}}\zeta^{g(\beta)}_p & \mbox{ if }\; m\mbox{ even, or } m
\mbox{ odd and } p\equiv1~(\mbox{mod }4),\\
\pm\sqrt{-1}p^{\frac{m}{2}}\zeta^{g(\beta)}_p  & \mbox{ if }m \mbox{ odd and }
p\equiv3~(\mbox{mod }4)
\end{array}\right.
\end{align}
for some $p$-ary function $g:\mathbb{F}_{q}\rightarrow\mathbb{F}_p$, which is
called the {\it associated function of} $f$.
A $p$-ary bent function $f$ is \emph{weakly regular} if there
is a complex number $u$ with unit magnitude such that $W_f (\beta) = u p^{\frac{m}{2}}\zeta^{\tilde{f}(\beta)}_p$
for some $p$-ary function $\tilde{f}:\mathbb{F}_{q}\rightarrow\mathbb{F}_p$.
In this case, we call $\tilde{f}$ the \emph{dual} of $f$.
In particular, when $u=1$, we say that $f$ is \emph{regular} $p$-ary bent.
We call $\epsilon$ the \emph{sign} of $W_f$ of a weakly regular $p$-ary bent function $f$
when $W_f (\beta) = \epsilon\eta(-1)^{m/2} p^{\frac{m}{2}}\zeta^{\tilde{f}(\beta)}_p$
for every $\beta\in\mathbb{F}_q$ and $\eta(a) = a^{(p-1)/2}$.

\subsection{Fourier-reflexive partitions}$~$

It is well known that every additive character of $\mathbb{F}_{q}$ can be expressed as
$$\chi_{a}(x)=\zeta_{p}^{\Tr(ax)},\ x\in \mathbb{F}_{q}$$
for some $a$ in $\mathbb{F}_{q}$.
In particular, we call $\chi_0$ the {\it trivial additive character} and $\chi_1$ the {\it canonical additive character} of $\mathbb{F}_{q}$. The orthogonality relation of additive characters (see \cite{LN}) is given by
$$ \sum_{x\in \mathbb{F}_{q}}\chi_a(x)=q\delta_{0,a},$$
where $\delta$ is the Kronecker delta function.
For a subset $D$ of $\mathbb{F}_q$ and $\beta \in \mathbb{F}_q$, we define
\[
	\chi_{\beta}(D) := \sum_{\alpha\in D}\zeta^{\operatorname{Tr}(\beta\alpha)}_p.
\]


We say that $P=\{P_i\}_{i=0}^d$ is a \emph{partition} of a set $X$ if all the blocks $P_i$ are pairwise disjoint and cover $X$. We write $|P|$ for the number of blocks in the partition $P$. Two partitions $P$ and $Q$ of a set X are
called {\it identical} if $|P| = |Q|$ and each block of two partitions coincides after suitable reordering.

Zinoviev and  Ericson in \cite{ZE2} proved the following result.

\begin{lem} \cite{ZE2} \label{eqn:lem2.2}
Let  $P=\{P_i\}_{i=0}^d$  be a partition of $\Bbb F_{q}$,
where $P_0=\{0\}$,
and let $R_i$ be a partition of $\Bbb F_{q}\times\Bbb F_{q}$ defined by
\[
(\alpha,\beta)\in R_i \Longleftrightarrow \alpha-\beta\in P_i, \; \; i=0,1,\ldots,d.
\]
Then $(\Bbb F_{q}, \{R_l\})_{l=0}^d$ is a symmetric association scheme on $\Bbb F_{q}$
if and only if there is another partition  $Q=\{Q_i\}_{i=0}^d$ of $\Bbb F_{q}$  such that for any $i,j\in\{0,1,\ldots,d\},$

$(1)$ the sum $\chi_{\alpha}(P_i)$
only depends on  the number $j$ such that  $\alpha\in Q_j$ and

$(2)$ the sum $\chi_{\beta}(Q_j)$
only depends on the number $i$ such that  $\beta\in P_i$.
\end{lem}

In this case, one can show that the partition $Q$ is unique up to permutation on $\{0,1,\ldots,d\}$. The partition $P$ in Lemma 2.2 is a so-called \emph{$B$-partition} \cite{ZE2}. Moreover, if $P=Q$, then the partition is also called $F$-{\it partition} \cite{ZE1}.

We define Fourier-reflexive partitions introduced by Gluesing-Luerssen \cite{GL} as follows.

\begin{defn} Let  $P=\{P_i\}_{i=0}^d$ be a partition of $G=\mathbb{F}_q$. The dual partition, denoted by $\widehat{P}$, is the partition of the character group $\widehat{G}$ defined by the equivalence relation
$$ \chi\sim_{\widehat{P}} \chi' \Longleftrightarrow \chi (P_i)=\chi'(P_i) \mbox{ for all } i=0, \ldots, d,$$ where $\chi$ and $\chi'$ are  characters of $\widehat{G}$. By identifying the group $G$ with its character group $\widehat{G}$ via $\beta\mapsto\chi_{\beta}$, the equivalence relation is restated as follows:
$$ \beta\sim_{\widehat{P}} \beta' \Longleftrightarrow \chi_{\beta} (P_i)=\chi_{\beta'}(P_i) \mbox{ for all } i=0, \ldots, d,$$
where $\chi_{\beta}$ and $\chi_{\beta'}$ are characters of $\widehat{G}$.
The partition $P$ is called \emph{Fourier-reflexive} if $\widehat{\widehat{P}}=P$.
\end{defn}

We point out that if $\widehat{P}=\{Q_j\}_{j=0}^{d'}$ is the dual partition of $P$, then $Q_j=\{0\}$ for some $j=0,1,\ldots,d'$. Indeed,
let $0\sim_{\widehat{P}} \beta$. Then $\chi_{\beta}(P_i)=|P_i|$ for all $i=0,1,\ldots,d$; therefore, we have $q=\sum_{i=0}^d|P_i|=\sum_{i=0}^d\chi_{\beta}(P_i)=\sum_{x\in\mathbb{F}_q}\zeta^{\Tr(\beta x)}_p=q\delta_{0,\beta}$. Thus we get $\beta=0$.
From this observation, we may always assume that $Q_0=\{0\}$.

We also point out that the blocks $Q_j$ of $\widehat{P}=\{Q_j\}_{j=0}^{d'}$ become
\[
Q_j=\{\beta\in \mathbb{F}_q:\chi_{\beta}(P_i)=c_{i,j}\text{ for all }i=0,\ldots,d\},
\]
where $c_{i,j}$ are complex numbers.

\begin{lem}\cite[Theorem 2.4]{GL}  \label{eqn:lem2.4}
A partition $P=\{P_i\}_{i=0}^d$ of $\mathbb{F}_q$ is Fourier-reflexive if and only if $|P|=|\widehat{P}|$.
\end{lem}

By using Lemmas \ref{eqn:lem2.2} and \ref{eqn:lem2.4}, we have the following result.

\begin{thm} \cite{GL,ZE2} \label{eqn:thm2.5}
Let  $P=\{P_i\}_{i=0}^d$ be a partition of $\Bbb F_{q}$,
where $P_0=\{0\}$, and let $\widehat{P}=\{Q_j\}_{j=0}^{d'}$ be the dual partition of $P$, where $Q_0=\{0\}$.
Let $R_i$ be a partition of $\Bbb F_{q}\times\Bbb F_{q}$ defined by
\[
(\alpha,\beta)\in R_i \Longleftrightarrow \alpha-\beta\in P_i, \; \; i=0,1,\ldots,d.
\]
A necessary and sufficient condition such that $(\Bbb F_{q}, \{R_i\})_{i=0}^d$ is a symmetric association scheme on $\Bbb F_{q}$ is $d=d'$.

 \end{thm}

\section{A main result}

In this section, we discuss the main result of this paper. We find an explicit criterion
for $p$-ary functions to produce association schemes in terms of their Walsh spectrum.




\begin{lem} \label{eqn:lem3.1}
 Let $f$ be a $p$-ary function. For a given $\beta\in\mathbb{F}_q$, let
 \[
 N^f_{i,j}(\beta)=|\{x\in\mathbb{F}_q:f(x)=i,~\Tr(\beta x)=j\}|.
 \]
 Then we have
 \[
 p^2N^f_{i,j}(\beta)=-q+p|D_{f,i}|+p\sum_{x\in \Bbb F_q}\delta_{j,\Tr(\beta x)}+\sum_{y\in\mathbb{F}^*_p}\left(\zeta_p^{-i}\sum_{z\in\mathbb{F}^*_p}\zeta_p^{jz}W_f(z \beta)\right)^{\sigma_{y}},
 \]
 where $\sigma_y$ is an automorphism of $\mathbb{Q}(\zeta_p)$ over $\mathbb{Q}$
such that $\sigma_y(\zeta_p)=\zeta_p^y$ for $y\in\mathbb{F}_p^*$.
\end{lem}

\begin{proof}
We have
\begin{multline}\label{eq1}
 p^2N^f_{i,j}(\beta)=\sum_{y,z\in \Bbb F_p} \sum_{x\in \Bbb F_q} \zeta_p^{y(f(x)-i)+z(\Tr(\beta x)-j)}
=\sum_{x\in\mathbb{F}_q}(1+\sum_{y\in\mathbb{F}^*_p}\zeta^{y(f(x)-i)}_p)(1+\sum_{z\in\mathbb{F}^*_p}\zeta^{z(\Tr(\beta x)-j}_p))\\
=q+\sum_{x\in \Bbb F_q} \sum_{y\in \Bbb F_p^*} \zeta_p^{y(f(x)-i)}+\sum_{x\in \Bbb F_q}\sum_{z\in \Bbb F_p^*}  \zeta_p^{z(\Tr(\beta x)-j)}+\sum_{y,z\in\mathbb{F}^*_p}\sum_{x\in\mathbb{F}_q}\zeta^{y(f(x)-i)+z(\Tr(\beta x)-j)}_p.
\end{multline}
The first two double sums in Eq.~\eqref{eq1} are computed as follows:
$$\sum_{x\in \Bbb F_q} \sum_{y\in \Bbb F_p^*} \zeta_p^{y(f(x)-i)}=\sum_{x\in \Bbb F_q}(-1+p\delta_{i,f(x)})=-q+p|D_{f,i}|$$ and $$\sum_{x\in \Bbb F_q}\sum_{z\in \Bbb F_p^*}  \zeta_p^{z(\Tr(\beta x)-j)}=\sum_{x\in \Bbb F_q}(-1+p\delta_{j,\Tr(\beta x)})=-q+p\sum_{x\in \Bbb F_q}\delta_{j,\Tr(\beta x)}.$$
The last double sum $\sum_{y,z\in\mathbb{F}^*_p}\sum_{x\in\mathbb{F}_q}\zeta^{y(f(x)-i)+z(\Tr(\beta x)-j)}_p$
in Eq.~\eqref{eq1} is simplified as the following:
\begin{eqnarray*}&&
\sum_{y,z\in\mathbb{F}^*_p}\zeta_p^{-iy-jz} \sum_{x\in\mathbb{F}_q}\zeta^{yf(x)+z\Tr(\beta x)}_p=\sum_{y,z\in\mathbb{F}^*_p}\zeta_p^{-iy-jz} \sum_{x\in\mathbb{F}_q}(\zeta^{f(x)+\Tr(\frac zy\beta x)}_p)^{\sigma_{y}}\\
&=&\sum_{y,z\in\mathbb{F}^*_p}\zeta_p^{-iy-jz}W_f(-\frac zy \beta)^{\sigma_{y}}
=\sum_{y\in\mathbb{F}^*_p}(\sum_{z\in\mathbb{F}^*_p}\zeta_p^{-i-j\frac zy}W_f(-\frac zy \beta))^{\sigma_{y}}\\
&=&\sum_{y\in\mathbb{F}^*_p}(\zeta_p^{-i}\sum_{z\in\mathbb{F}^*_p}\zeta_p^{jz}W_f(z \beta))^{\sigma_{y}}.
\end{eqnarray*}
The result follows immediately by combining the computation results obtained as above all together.
\end{proof}

\begin{cor}\label{eqn:cor3.2}
 Let $f$ be a $p$-ary function.
 Then for $\beta\in\mathbb{F}_q$ and $i\in\mathbb{F}_p$, we have
 \[
 \chi_{\beta}(D_{f,i})=\frac{1}{p^2}\sum_{j\in\mathbb{F}_p}\zeta^j_p\sum_{y\in\mathbb{F}^*_p}\left(\zeta_p^{-i}\sum_{z\in\mathbb{F}^*_p}\zeta_p^{jz}W_f(z \beta)\right)^{\sigma_{y}}.
 \]
\end{cor}

\begin{proof}
The result follows from the fact that $\chi_{\beta}(D_{f,i})=\sum_{j=0}^{p-1}N^f_{i,j}(\beta)\zeta^j_p$ and using Lemma \ref{eqn:lem3.1}.
\end{proof}

The following proposition plays an important role in proving our main result.

\begin{prop}\label{eqn:prop3.3}
Let $f$ be a $p$-ary function and $f(\mathbb{F}_p)$ be the image of $\mathbb{F}_p$ under $f$. Then for any two distinct elements $\beta, \beta'\in\Bbb F_q^*$, we have
\[
\chi_{\beta}(D_{f,i})=\chi_{\beta'}(D_{f,i})\text{ for all } i\in f(\mathbb{F}_p)\iff
W_f(z\beta)=W_f(z\beta') \text{ for all } z\in\mathbb{F}^*_p.
\]
\end{prop}

\begin{proof}
Let us put $I=f(\mathbb{F}_p)$.

$(\Rightarrow)$ The Walsh-Hadamard transform of a $p$-ary function $f$ is expressed by
\[
W_f(\beta)=\sum_{i\in I}\chi_{-\beta}(D_{f,i})\zeta^i_p.
\]
By the assumption, for $z\in\mathbb{F}^*_p$ we have
$
\chi_{-z\beta}(D_{f,i})=\chi_{-\beta}(D_{f,i})^{\sigma_z}=\chi_{-\beta'}(D_{f,i})^{\sigma_z}=\chi_{-z\beta'}(D_{f,i})
$
for all $i\in I$. Thus for $z\in\mathbb{F}^*_p$ we have
\[
W_f(z\beta)=\sum_{i\in I}\chi_{-z\beta}(D_{f,i})\zeta^i_p
=\sum_{i\in I}\chi_{-z\beta'}(D_{f,i})\zeta^i_p=W_f(z\beta').
\]
$(\Leftarrow)$
By Lemma \ref{eqn:lem3.1} and Corollary \ref{eqn:cor3.2}, we have
\begin{eqnarray*}
&& W_f(z \beta)=W_f(z \beta') \mbox{ for all } z\in \mathbb{F}^*_p\\
&\Rightarrow& \sum_{z\in\mathbb{F}^*_p}(\zeta_p^{jz}-1)W_f(z \beta)=\sum_{z\in\mathbb{F}^*_p}(\zeta_p^{jz}-1)W_f(z \beta') \mbox{ for all } j\in \mathbb{F}_p\\
&\Rightarrow& \sum_{y\in\mathbb{F}^*_p}\left(\zeta_p^{-i}\sum_{z\in\mathbb{F}^*_p}(\zeta_p^{jz}-1)W_f(z \beta)\right)^{\sigma_{y}}=\sum_{y\in\mathbb{F}^*_p}\left(\zeta_p^{-i}\sum_{z\in\mathbb{F}^*_p}(\zeta_p^{jz}-1)W_f(z \beta')\right)^{\sigma_{y}}\\
&&\mbox{ for all } i\in I\text{ and }j\in \mathbb{F}_p\\
&\Rightarrow& N^f_{i,j}(\beta)- N^f_{i,0}(\beta)= N^f_{i,j}(\beta')- N^f_{i,0}(\beta') \text{ for all }i\in I\text{ and } j\in \mathbb{F}_p
\\
&\Rightarrow&\chi_{\beta}(D_{f,i})=\chi_{\beta'}(D_{f,i})\mbox{ for all } i\in I.
\end{eqnarray*}
The last implication follows from the fact that $\{\zeta_p,\ldots,\zeta^{p-1}\}$ is a basis for $\mathbb{Q}(\zeta_p)$ over $\mathbb{Q}$.
This completes the proof.
\end{proof}

We denote by $S^*$ the set of nonzero elements of a subset $S$ of $\mathbb{F}_q$.

Let $f$ be a $p$-ary function with $f(0)=0$.
Then $P=\{D^*_{f,i}\}_{i\in f(\mathbb{F}_q^*)}$ is a partition of $\Bbb F^*_{q}$.
Let $R_i$ be a partition of $\Bbb F_{q}\times\Bbb F_{q}$ defined by
\begin{align}\label{eqn:(3.1)}
R_{-1}=\{(\beta,\beta):\beta\in\mathbb{F}_q\},\text{ and }
(\alpha,\beta)\in R_i \Longleftrightarrow \alpha-\beta\in D^*_{f,i}
\text{ for } i\in f(\mathbb{F}_q^*).
\end{align}
If $(\mathbb{F}_q,\{R_i\})_{ i\in f(\mathbb{F}_q^*)})$, where $R_{-1}$ is deleted,
forms a symmetric $d$-class association scheme, then we say that the partition $P=\{D^*_{f,i}\}_{ i\in f(\mathbb{F}_q^*)}$ of $\mathbb{F}^*_q$ {\it induces an $d$-class association scheme} $(\mathbb{F}_q,\{R_i\}_{ i\in f(\mathbb{F}_q^*)})$.

\begin{thm} \label{eqn:thm3.4}
Let $f$ be a $p$-ary function with $f(0)=0$ and let $I=f(\mathbb{F}_q^*)$. Then the following statements are equivalent.

$(i)$ The partition $P=\{D^*_{f,i}\}_{ i\in I}$ of $\mathbb{F}^*_q$ induces an $|I|$-class association scheme $(\mathbb{F}_q,\{R_i\}_{ i\in I})$ defined in (\ref{eqn:(3.1)});

$(ii)$ The sizes of $I$ and $\{(W_{f}(\beta), W_{f}(2\beta), \ldots, W_{f}((p-1)\beta)):\beta\in \Bbb F^*_q\}$ are same.
\end{thm}

\begin{proof}

Let $V(\beta)=(W_{f}(\beta), W_{f}(2\beta), \ldots, W_{f}((p-1)\beta))$ be a vector. By Proposition \ref{eqn:prop3.3}, we have $$\chi_{\beta}(D_{f,i})=\chi_{\beta'}(D_{f,i})\text{ for all } i\in f(\mathbb{F}_p)\iff V(\beta)=V(\beta')$$
for two distinct elements $\beta$ and $\beta'\in \Bbb F_q^*$.
Then the result follows from Theorem \ref{eqn:thm2.5}.
\end{proof}




We give the following example for illustration of Theorem \ref{eqn:thm3.4}.

\begin{exa}

Let $p=3$, $q=27$, and $f(x)=\Tr(2x-x^5)$.  Using Magma,  we get the following facts:

(1) The set of the Walsh-Hadamard values of $f$ at nonzero elements is $\{0, 9, -9\}$.



(2) The set $\{(W_f(\beta), W_f(2\beta)): \beta\in \Bbb F_{27}^*\}$ is equal to $\{(-9,0), (0,-9), (0,0), (0,9), (9,0)\}$.

By Theorem \ref{eqn:thm3.4}, we can see that the partition
$\{ D^*_{f,i}\}_{i\in f(\Bbb F_q^*)}$ of $\mathbb{F}^*_q$ cannot induce a $3$-class association scheme.

\end{exa}

\section{Applications of a main result}

In this section we find various applications of the main result in Section 3.
We first discuss a characterization of $p$-ary bent functions in terms of association schemes in the subsection~\ref{subsection4.1}.
In the subsection~\ref{subsection4.2},
we use $p$-ary functions satisfying that $f(ax)=f(x)$ for all $(a,x)\in\mathbb{F}^*_p\times\mathbb{F}_q$ for construction of infinite families of a few-class association schemes.
Especially, we use some specific $p$-ary functions studied in \cite{WYL}.
Furthermore, we find four classes of $p$-ary two-weight linear codes,
which are constructed from the association schemes produced in the subsection~\ref{subsection4.2}.

\subsection{Characterization of $p$-ary bent functions}$~$\label{subsection4.1}


\begin{lem}\label{eqn:lem4.2}
Let $f$ be a $p$-ary bent function with $f(0)=0$ and let $g$ be the associated $p$-ary function of $f$. Then following statements are true.

$(i)$ $f$ is surjective;

$(ii)$ the Walsh-Hadamard transform of $f$ can be written as
\begin{align*}
W_f(\beta)=\mu(\beta)(p^*)^{\frac{m}{2}}\zeta^{g(\beta)}_p,
\end{align*}
where $\mu(\beta)=\pm1$ for all $\beta\in\mathbb{F}_q$ and $p^*=(-1)^{(p-1)/2}p$ and $g$ is the associated $p$-ary function of $f$;

$(iii)$ if $m$ is even, then \[
\sum_{\beta\in D_{g,i}}\mu(\beta)-\sum_{\beta\in D_{g,0}}\mu(\beta)=-p^{\frac{m}{2}}\text{ for all } i\in\mathbb{F}^*_p
\]
and if $m$ is odd, then
\begin{align*}
\sum_{\beta\in D_{g,i}}\mu(\beta)-\sum_{\beta\in D_{g,0}}\mu(\beta)=
\left\{
\begin{array}{c l}
p^{\frac{m-1}{2}} & \mbox{ if } i\in S,\\
-p^{\frac{m-1}{2}} & \mbox{ if } i\in NS,
\end{array}\right.
\end{align*}
where $S$ (resp., $NS$) is the set of squares (resp., non-squares) in $\mathbb{F}^*_p$.
\end{lem}

\begin{proof}
$(i)$ Assume that there exists $i\in \Bbb F_p$ such that $D_{f,i}=\emptyset.$ By Lemma 3.1, we have $N_{i,0}^f(0)=0$,  $|D_{f,i}|=0$, and $\sum_{y\in\mathbb{F}^*_p}\left(\zeta_p^{-i}W_f(0)\right)^{\sigma_{y}}=-q.$ Note that $|W_f(0)|=\sqrt{q}$; therefore, we have $$\bigg|\sum_{y\in\mathbb{F}^*_p}\left(\zeta_p^{-i}W_f(0)\right)^{\sigma_{y}}\bigg|\le (p-1) \sqrt{q}<q,$$ which is a contradiction. This shows that $f$ is surjective.


$(ii)$ This is a direct consequence of Eq. (2.1).

$(iii)$ Case $1$: $m$ is even.

By using $1=-\sum_{i=1}^{p-1}\zeta^i_p$ and $q\zeta^{f(0)}_p=\sum_{\beta\in\mathbb{F}_q}W_f(\beta)$ for any $p$-ary function $f$, we have
\begin{multline*}
-\sum_{i=1}^{p-1}p^m\zeta^i_p=p^m=\sum_{\beta\in\mathbb{F}_q}W_f(\beta)=p^{\frac{m}{2}}\sum_{\beta\in\mathbb{F}_q}\mu(\beta)\zeta^{g(\beta)}_p\\
=p^{\frac{m}{2}}\sum_{i=0}^{p-1}\sum_{\beta\in D_{g,i}}\mu(\beta)\zeta^i_p=p^{\frac{m}{2}}\sum_{i=1}^{p-1}\left(\sum_{\beta\in D_{g,i}}\mu(\beta)-\sum_{\beta\in D_{g,0}}\mu(\beta)\right)\zeta^i_p.
\end{multline*}
The result follows from the fact that $\{\zeta_p,\ldots,\zeta^{p-1}_p\}$ is a basis for $\mathbb{Q}(\zeta_p)$ over $\mathbb{Q}$.

Case $2$: $m$ is odd.

We note that
$\sum_{i=0}^{p-1}\zeta_p^{i^2}=\sqrt{p^*},$ where $p^*=(-1)^{(p-1)/2}p$.

We then have
\[
p^m=\sum_{\beta\in\mathbb{F}_q}W_f(\beta)=(p^*)^{\frac{m}{2}}\sum_{\beta\in\mathbb{F}_q}\mu(\beta)\zeta^{g(\beta)}_p;
\]
thus we get
\[
p^{\frac{m-1}{2}}\sum_{i=0}^{p-1}\zeta_p^{i^2}=p^m(p^*)^{-\frac{m}{2}}=\sum_{i=0}^{p-1}\sum_{\beta\in D_{g,i}}\mu(\beta)\zeta_p^{j}=\sum_{i=1}^{p-1}\left(\sum_{\beta\in D_{g,i}}\mu(\beta)-\sum_{\beta\in D_{g,0}}\mu(\beta)\right)\zeta_p^{j}.
\]
We can see that $\sum_{i=0}^{p-1}\zeta_p^{i^2}=1+2\sum_{i\in S}\zeta_p^{i}=
-\sum_{i=1}^{p-1}\zeta^i_p+2\sum_{i\in S}\zeta_p^{i}=\sum_{i\in S}\zeta_p^{i}-\sum_{i\in NS}\zeta_p^{i}$; this implies that
\[
p^{\frac{m-1}{2}}\left(\sum_{i\in S}\zeta_p^{i}-\sum_{i\in NS}\zeta_p^{i}\right)
=\sum_{i=1}^{p-1}\left(\sum_{\beta\in D_{g,i}}\mu(\beta)-\sum_{\beta\in D_{g,0}}\mu(\beta)\right)\zeta_p^{j},
\]
and the result follows.
\end{proof}

\begin{thm}  Let $p$ be an odd prime number, and let $f$ be a  $p$-ary bent function with $f(0)=0$.  The partition $P=\{D^*_{f,i}\}_{ i\in f(\mathbb{F}_q^*)}$ induces a $p$-class association scheme $(\mathbb{F}_q,\{R_i\}_{ i\in f(\mathbb{F}_q^*)})$ defined in Equation (\ref{eqn:(3.1)}) if and only if $f$ is weakly regular.
\end{thm}
\begin{proof}
Let us put $A_+=\{\beta\in\mathbb{F}_q:\mu(\beta)=1\}$ and $A_-=\{\beta\in\mathbb{F}_q:\mu(\beta)=-1\}$.

$(\Rightarrow)$
By Lemma \ref{eqn:lem4.2}, we have $|{f(\mathbb{F}_q)}|=p$.
Let $\widehat{P}=\{Q_j\}_{j=0}^{p-1}$ be the dual partition of $P$.
We have $Q_{-1}=\{0\}$.
By Proposition \ref{eqn:prop3.3}, we get that $\beta$ and $\beta'$ in $\mathbb{F}_q$ are in the same block of $\widehat{P}$ $\iff$ $W_f(z\beta)=W_f(z\beta')$ for all $z\in\mathbb{F}^*_p$ $\iff$ $g(z\beta)=g(z\beta')$ and $\mu(z\beta)=\mu(z\beta')$ for all $z\in\mathbb{F}^*_p$. This shows that
\begin{align}\label{eqn:4.1}
   \text{either }Q_j\subseteq A_+ \; \text{ or } \; Q_j\subseteq A_-\text{ for } j\in\mathbb{F}_p.
\end{align}

Case 1: $m$ is even. We divide this case into two subcases as follows.

Case 1-1: assume $D_{g,0}\subseteq A_+$. By Lemma \ref{eqn:lem4.2}, we have
\begin{align}\label{eqn:4.2}
\sum_{\beta\in D_{g,i}}\mu(\beta)=|D_{g,0}|-p^{\frac{m}{2}} \text{ for  all } i\in\mathbb{F}^*_p.
\end{align}
We claim that $D_{g,i}\subseteq A_+$ for all $i\in\mathbb{F}^*_p$.
Assume not, that is, $D_{g,i}$ is not contained in $ A_+$ for some $i\in\mathbb{F}^*_p$, for contradiction.
Then by Eq. (\ref{eqn:4.1}), it is enough to consider the following two cases:

$(i)$ $D_{g,i}\subseteq A_-$ for all $i\in\mathbb{F}^*_p$, or

$(ii)$ $Q_j\subseteq A_+$ and $Q_l\subseteq A_-$ for some $j,l\in\mathbb{F}^*_p$ with $j\neq l$.

Assume that $(i)$ is true. By Eq. (\ref{eqn:4.2}), we have $|D_{g,i}|+|D_{g,0}|=p^{\frac{m}{2}}$ for all $i\in\mathbb{F}^*_p$. Adding all these equations together with $|D_{g,0}|+|D_{g,0}|=2|D_{g,0}|$, we have that $q+p|D_{g,0}|=(p-1)p^{\frac{m}{2}}+2|D_{g,0}|$ \; or \; $(p-2)|D_{g,0}|=-q+(p-1)p^{\frac{m}{2}}<0$,
which is impossible.\\
Assume that $(ii)$ is true. By Eq. (\ref{eqn:4.2}), we have that $|D_{g,j}|=|D_{g,0}|-p^{\frac{m}{2}}\geq0$ and $|D_{g,l}|=-|D_{g,0}|+p^{\frac{m}{2}}\geq0$; thus, $|D_{g,0}|=p^{\frac{m}{2}}$.
By Eq. $(\ref{eqn:4.2})$, we have $\sum_{\beta\in D_{g,i}}\mu(\beta)=0$ for all $i\in\mathbb{F}^*_p$, and so $|D_{g,i}|=0$ for all $\mathbb{F}^*_p$ by Eq. (\ref{eqn:4.1}), which is impossible because $p^{\frac{m}{2}}=\sum_{i\in\mathbb{F}_p}|D_{g,i}|=q$.
This proves our claim; therefore, $f$ is weakly regular bent.

Case 1-2: assume $D_{g,0}\subseteq A_-$. By Lemma \ref{eqn:lem4.2}, we have
\begin{align}\label{eqn:4.3}
\sum_{\beta\in D_{g,i}}\mu(\beta)=-|D_{g,0}|-p^{\frac{m}{2}}<0 \text{ for all } i\in\mathbb{F}^*_p.
\end{align}
By (\ref{eqn:4.1}) and (\ref{eqn:4.3}), we have $D_{g,i}\subseteq A_{-}$ for all $i\in\mathbb{F}^*_p$, and thus $f$ is weakly regular.

Case 2: $m$ is odd. We also divide this case into two subcases as follows.

Case 2-1: assume $D_{g,0}\subseteq A_+$. By Lemma \ref{eqn:lem4.2}, we have
\begin{align}\label{eqn:4.4}
\sum_{\beta\in D_{g,i}}\mu(\beta)=
\left\{
\begin{array}{c l}
|D_{g,0}|+p^{\frac{m-1}{2}} & \mbox{ if } i\in S,\\
|D_{g,0}|-p^{\frac{m-1}{2}} & \mbox{ if } i\in NS,
\end{array}\right.
\end{align}
We claim that $D_{g,i}\subseteq A_+$ for all $i\in\mathbb{F}^*_p$.
Assume not, that is, $D_{g,i}$ is not contained in $ A_+$ for some $i\in\mathbb{F}^*_p$, for contradiction.
Since $D_{g,i}\subseteq A_+$ for $i\in S$ by Eqs. (\ref{eqn:4.1}) and (4.4),
it is sufficient to consider the following two cases:

$(i)$ $D_{g,i}\subseteq A_-$ for all $i\in NS$, or

$(ii)$ $Q_j\subseteq A_+$ and $Q_l\subseteq A_-$ for some $j,l\in NS$ with $j\neq l$.

Assume that $(i)$ is true. By Eq. (4.4), we have $|D_{g,i}|=|D_{g,0}|+p^{\frac{m-1}{2}}$
for $i\in S$ and $|D_{g,i}|=-|D_{g,0}|+p^{\frac{m-1}{2}}$ for $i\in NS$.
Adding all these equations together with $|D_{g,0}|=|D_{g,0}|$ leads to $q=|D_{g,0}|+2\frac{p-1}{2}p\frac{m-1}{2}$,
and so $|D_{g,0}|=q-(p-1)p\frac{m-1}{2}$. We compute $|D_{g,i}|$ for $i\in NS$ as follows: $|D_{g,i}|=-|D_{g,0}|+p^{\frac{m-1}{2}}=-q+p^{\frac{m+1}{2}}<0$, which is impossible.

Assume that $(ii)$ is true. By the same argument as Case 1-1 $(ii)$,
we have a contradiction. This proves that $f$ is weakly regular bent.

Case 2-2: assume $D_{g,0}\subseteq A_-$.
By Lemma \ref{eqn:lem4.2}, we have
\begin{align}\label{eqn:4.5}
\sum_{\beta\in D_{g,i}}\mu(\beta)=
\left\{
\begin{array}{c l}
-|D_{g,0}|+p^{\frac{m-1}{2}} & \mbox{ if } i\in S,\\
-|D_{g,0}|-p^{\frac{m-1}{2}} & \mbox{ if } i\in NS.
\end{array}\right.
\end{align}
To show that $D_{g,i}\subseteq A_-$ for all $i\in\mathbb{F}^*_p$,
assume not, for contradiction.
Since $D_{g,i}\subseteq A_-$ for $i\in NS$ by Eqs. (\ref{eqn:4.1}) and $(4.5)$,
it is enough to consider the following two cases:

$(i)$ $D_{g,i}\subseteq A_+$ for all $i\in S$, or

$(ii)$ $Q_j\subseteq  A_+$ and $Q_l\subseteq A_-$ for some $j,l\in S$ with $j\neq l$.

By the same argument as Case 2-1, we also get a contradiction.
This proves that $f$ is weakly regular bent.

$(\Leftarrow)$
Since $f$ is weakly regular bent, we have $W_f (\beta) = \epsilon\eta(-1)^{m/2} p^{\frac{m}{2}}\zeta^{\tilde{f}(\beta)}_p$ and
$\tilde{f}$ is also weakly regular bent. Then $\tilde{f}$ is surjective by Lemma \ref{eqn:lem4.2}; thus, $|P|=|\widehat{P}|$ by Proposition \ref{eqn:prop3.3}. The result follows from Theorem  \ref{eqn:thm3.4}.
\end{proof}

\begin{rem}{\rm
In \cite[Theorem 3]{PTFL}, they proved that if $f$ is a
weakly regular $p$-ary bent function satisfying that $f(ax)=a^l f(x)$
for all $a\in\mathbb{F}^*_p, x\in\mathbb{F}_q$ with $(l-1,p-1)=1$ and $f(0)=0$, then the partition $P=\{D^*_{f,i}\}_{i=0}^{p-1}$ of $\mathbb{F}^*_q$ induces a $p$-class association scheme $(\mathbb{F}_q,\{R_i\}_{ i=0}^{p-1})$, which is a special case of Theorem 4.2.
}
\end{rem}

\subsection{Construction of  association schemes}$~$\label{subsection4.2}

\begin{lem}\label{eqn:lem4.1}
 Let $f$ be a $p$-ary function. Then the following statements are equivalent.

 $(i)$ $f(ax)=f(x)$ for all $(a,x)\in\mathbb{F}^*_p\times\mathbb{F}_q$;

 $(ii)$ $W_f(a\beta)=W_f(\beta)$ for all $(a,\beta)\in\mathbb{F}^*_p\times\mathbb{F}_q$;

 $(iii)$ $a D_{f,i}=D_{f,i}$ for all $(a,i)\in\mathbb{F}^*_p\times\mathbb{F}_p$;

 $(iv)$ $\chi_{\beta}(D_{f,i})$ are rational integers for all $(i,\beta)\in\mathbb{F}_p\times\mathbb{F}_q$.
 \end{lem}

\begin{proof}
$(i)\Rightarrow(ii)$. For $a\in\mathbb{F}^*_p$, we have
\[
W_f(a\beta)=\sum_{x\in\mathbb{F}_q}\zeta_p^{f(x)-\Tr(a\beta x)}=
\sum_{x\in\mathbb{F}_q}\zeta_p^{f(a^{-1}x)-\Tr(\beta x)}=\sum_{x\in\mathbb{F}_q}\zeta_p^{f(x)-\Tr(\beta x)}=W_f(\beta).
\]

$(ii)\Rightarrow(i)$.
By using the inverse Walsh-Hadamard transform,  for  $a\in\mathbb{F}^*_p$
\[
p^m\zeta^{f(x)}_p=\sum_{\beta\in\mathbb{F}_q}W_f(\beta)\zeta^{\Tr(\beta x)}_p
=\sum_{\beta\in\mathbb{F}_q}W_f(a^{-1}\beta)\zeta^{\Tr(\beta x)}_p\\
=\sum_{\beta\in\mathbb{F}_q}W_f(\beta)\zeta^{\Tr(a\beta x)}_p
=p^m\zeta^{f(ax)}_p.
\]

$(iii)\Rightarrow(ii)$.
By using $W_f(\beta)=\sum_{i=0}^{p-1}\chi_{\beta}( D_{f,i})\zeta^i_p$ and $\chi_{\beta}(a D_{f,i})=\chi_{a\beta}(D_{f,i})$ for $a\in\mathbb{F}^*_p$, we get the result.

$(i)\Rightarrow(iii)$. It is straightforward.

$(iii)\Rightarrow(iv)$. We have $\chi_{\beta}(D_{f,i})^{\sigma_a}=\chi_{\beta}(aD_{f,i})=\chi_{\beta}(D_{f,i})$ for all $a\in\mathbb{F}^*_p$, and the result follows.

$(iv)\Rightarrow(iiii).$ We have $\chi_{\beta}(aD_{f,i})=\chi_{\beta}(D_{f,i})^{\sigma_a}=\chi_{\beta}(D_{f,i})$ for all $\beta\in\mathbb{F}_q$. The result follows from Lemma IV.2 in \cite{HL}.
\end{proof}

By Theorem 3.4 and Lemma 4.1, we have the following corollary.

\begin{cor} Let $f$ be a $p$-ary function with $f(0)=0$ and  $f(ax)=f(x)$ for all $a\in \Bbb F_p^*$ and let $I=f(\mathbb{F}_q^*)$. Then the partition $P=\{D^*_{f,i}\}_{ i\in I}$ of $\mathbb{F}^*_q$ induces an $|I|$-class association scheme $(\mathbb{F}_q,\{R_i\}_{ i\in I})$ defined in Equation (\ref{eqn:(3.1)}) if and only if  the sizes of $I$ and $\{W_{f}(\beta):\beta\in \Bbb F^*_q\}$ are same.
\end{cor}

In \cite{WYL}, the authors presented three classes of monomial functions with three-valued Walsh spectrum. Using the results in \cite{WYL} and Corollary 4.5, we have the following propositions.

Let $\phi(n)$ be the Euler $\phi$-function, and recall that $f(\mathbb{F}_q^*)=\{f(x): x\in \Bbb F_q^*\}$.

\begin{prop}\label{eqn:prop4.5}
Assume that   $m$  is  a positive integer,  $r$ is  an odd  prime such that $p\ge 3$ is a primitive root modulo $r^m$. Let $q=p^{\phi(r^m)}$ and $f(x)=\Tr_{q/p}(x^{\frac{q-1}{r^m}})$.

$(i)$ Then $\{D^*_{f,i}\}_{i\in f(\mathbb{F}_q^*)}$ induces a 2-class  association scheme if and only if $|f(\mathbb{F}^*_q)|=2$; equivalently, $m=1$ or $r\equiv 1\pmod p$.

$(ii)$ Then $\{D^*_{f,i}\}_{i\in f(\mathbb{F}_q^*)}$ induces a 3-class association scheme if and only if $|f(\mathbb{F}^*_q)|=3$; equivalently,  $m>1$ and  $r\not\equiv 1\pmod p$.

\end{prop}

\begin{proof}
Let $\Bbb F_q^*=\langle \gamma \rangle$. Then $\Bbb F_p^*=\langle \gamma ^{\frac{q-1}{p-1}}\rangle$. By the assumption on $p$ and $r$, $\gcd(r,p-1)=1$ and $r^m \mid q-1=p^{\phi(r^m)}-1$, $r^m\mid \frac{q-1}{p-1}$. Hence, $f(ax)=f(a)$ for any $a\in \Bbb F_p^*$, and thus we can use Theorem \ref{eqn:thm3.4}. By the proof of \cite[Theorem 6]{WYL}, we have $f(\Bbb F_q^*)=\{\phi(r^m), -r^{m-1},0\}$, where the frequency of $0$ as the image of $f$ is $q-1-\frac{q-1}{r^{m-1}}$.
By \cite[Theorem 10]{WYL}, $$\{W_f(\beta):\beta\in \Bbb F_q^*\}=\{\sqrt{q}\zeta_p^{\phi(r^m)}+A, \sqrt{q}\zeta_p^{-r^{m-1}}+A, \sqrt{q}+A\},$$ where $A=1-\frac{\sqrt{q}+1}{r^m}(r^m-1+\zeta_p^{\phi(r^m)}+\phi(r)(\zeta_p^{-r^{m-1}}-1))$.
Since $\gcd(r,p)=1$, we have that $-r^{m-1}\not\equiv 0\pmod p$ and $\phi(r^m)=r^{m-1}(r-1)\not\equiv -r^{m-1}\pmod p$;
thus, $\phi(r^m)\equiv 0\pmod p$ if and only if $r\equiv 1\pmod p.$
Then the result follows from Corollary 4.5.
\end{proof}

\begin{exa}{\rm
(1) Let $p=3$, $m=1, r=5,$ $q=3^4$, and $f(x)=\Tr_{81/3}(x^{16})$. Then
$f(\mathbb{F}_q^*)=\{\phi(r^m), -r^{m-1}\}=\{1,2\}$, and it has size two. By Proposition \ref{eqn:prop4.5}, the partition $\{D^*_{f,i}\}_{i\in f(\mathbb{F}_q^*)}$ of $\mathbb{F}^*_q$ induces a 2-class association scheme.

(2) Let $p=3$, $m=2, r=7,$ $q=3^{42}$, and $f(x)=\Tr_{3^{42}/3}(x^{\frac{3^{42}-1}{49}})$.
Then $f(\mathbb{F}_q^*)=\{\phi(r^m), -r^{m-1},0\}=\{0,2\}$, which has size two. By Proposition \ref{eqn:prop4.5}, the partition $\{D^*_{f,i}\}_{i\in f(\mathbb{F}_q^*)}$ of $\mathbb{F}^*_q$ induces a 2-class association scheme.

(3) Let $p=3$, $m=2, r=5,$ $q=3^{20}$, and $f(x)=\Tr_{3^{20}/3}(x^{\frac{3^{20}-1}{25}})$.
Then $f(\mathbb{F}_q^*)=\{\phi(r^m), -r^{m-1},0\}=\Bbb F_3$, and it has size three. By Proposition \ref{eqn:prop4.5}, the partition $\{D^*_{f,i}\}_{i\in f(\mathbb{F}_q^*)}$ of $\mathbb{F}^*_q$ induces a 3-class association scheme.

}
\end{exa}


\begin{prop} \label{eqn:prop4.7}
Assume that   $m$  is  a positive integer,  $r\equiv1\pmod4$  is a  prime, and the order of a prime number $p$ modulo $r^m$ is $\frac{\phi(r^m)}{2}$.  Let $q=p^{\frac{\phi(r^m)}{2}}$ and $f(x)=\Tr_{q/p}(x^{\frac{q-1}{r^m}})$.

$(i)$ Then $\{D^*_{f,i}\}_{i\in f(\mathbb{F}_q^*)}$ induces a 2-class  association scheme if and only if $|f(\mathbb{F}^*_q)|=2$; equivalently,  $r\equiv 1\pmod p$.

$(ii)$ Then $\{D^*_{f,i}\}_{i\in f(\mathbb{F}_q^*)}$ induces a 3-class  association scheme if and only if $|f(\mathbb{F}^*_q)|=3$; equivalently, $r\not\equiv 1\pmod p$ and $m=1$.

$(iii)$ Then $\{D^*_{f,i}\}_{i\in f(\mathbb{F}_q^*)}$ induces a 4-class association scheme if and only if $|f(\mathbb{F}^*_q)|=4$; equivalently,   $r\not\equiv 1\pmod p$ and $m>1$.



\end{prop}

\begin{proof} Let $\Bbb F_q^*=\langle \gamma \rangle$. Then $\Bbb F_p^*=\langle \gamma ^{\frac{q-1}{p-1}}\rangle$.
 From the assumption on $r$ and $p$, we have that $\gcd(r,p-1)=1$ and $r^m \mid q-1=p^{\frac{\phi(r^m)}2}-1$; so, $r^m\mid \frac{q-1}{p-1}$. Hence, $f(ax)=f(a)$ for any $a\in \Bbb F_p^*$, and we can use Theorem \ref{eqn:thm3.4}.
 By the proof of \cite[Theorem 16]{WYL}, $$f(\Bbb F_q^*)=\bigg\{\frac{\phi(r^m)}2, \frac{(\sqrt r-1){r^{m-1}}}{2},\frac{(-\sqrt r-1){r^{m-1}}}{2},0\bigg\},$$ where the frequency of $0$ as the image of $f$ is $q-1-\frac{q-1}{r^{m-1}}$. By \cite[Theorem 20]{WYL}, $$\{W_f(\beta):\beta\in \Bbb F_q^*\}=\{\sqrt{q}\zeta_p^{\frac{\phi(r^m)}2}+B,
\sqrt{q}\zeta_p^{\frac{(\sqrt r-1){r^{m-1}}}{2}}+B,
\sqrt{q}\zeta_p^{\frac{(-\sqrt r-1){r^{m-1}}}{2}}+B,
\sqrt{q}+B\},$$ where $B=1-r^m-1+\zeta_p^{\frac{\phi(r^m)}2}+\frac{\phi(r)}2(\zeta_p^{\frac{(\sqrt r-1)r^{m-1}}{2}}-1)+\frac{\phi(r)}2(\zeta_p^{\frac{(-\sqrt r-1)r^{m-1}}{2}}-1)$.

(i) If $p=2$, then $r^{m-1}\not\equiv 0\pmod 2$ since $r\equiv 1\pmod 4$. In this case, we have $f(\Bbb F_q^*)=\{r^{m-1}\cdot\frac{r-1}{2}, r^{m-1}\cdot \frac{\sqrt r-1}2,  r^{m-1}\cdot \frac{-\sqrt r-1}2,0\}$, $\frac{r-1}{2}\equiv 0\pmod 2$, and $\frac{\sqrt r-1}2\not\equiv \frac{-\sqrt r-1}2\pmod 2$.
Hence we have $f(\Bbb F_q^*)=\{0,1\}$ when $p=2$.


(ii) If $p>2$, then $-\frac{r^{m-1}}2\not\equiv 0\pmod p$ and $\sqrt r-1\not\equiv -\sqrt r -1\pmod p$ since $\gcd(r,p)=1$. Then $f(\Bbb F_q^*)$ and the set $A=\{r-1, \sqrt r-1, -\sqrt r-1,0\} \pmod p$ have the same sizes.
Note that $r-1=\sqrt r-1 \iff \sqrt r=1$ and $r-1=-\sqrt r-1$ $\iff$ $\sqrt r=-1$. Hence, $|A|=2$ if and only if $r\equiv 1\pmod p$, $|A|=3$ if and only if $r\not\equiv 1\pmod p$ and $m=1$,  and $|A|=4$ if and only if $r\not\equiv 1\pmod p$ and $m>1$.

Therefore, the result follows from Corollary 3.5.
\end{proof}

\begin{exa}{\rm

(1) Let $p=2$, $m=1, r=17,$ $q=256$, and $f(x)=\Tr_{256/2}(x^{15})$.
Then $f(\Bbb F_q^*)=\{0,1\}$, which has size two. By Proposition \ref{eqn:prop4.7}, the partition $\{D^*_{f,i}\}_{i\in f(\mathbb{F}_q^*)}$  of $\mathbb{F}^*_q$ induces a 2-class association scheme.

(2) Let $p=19$, $m=1, r=5,$ $q=361$, and $f(x)=\Tr_{361/19}(x^{\frac{19^{10}-1}{25}})$.
Then $f(\Bbb F_q^*)=\{\frac{r-1}2, \frac{(\sqrt r-1){}}{2},\frac{(-\sqrt r-1){}}{2}\}=\{2,4,14\}$, and it has size three. By Proposition \ref{eqn:prop4.7}, the partition $\{D^*_{f,i}\}_{i\in f(\mathbb{F}_q^*)}$  of $\mathbb{F}^*_q$ induces a 3-class association scheme.

(3) Let $p=19$, $m=2, r=5,$ $q=19^{10}$, and $f(x)=\Tr_{19^{10}/19}(x^{\frac{19^{10}-1}{25}})$.
Then $f(\Bbb F_q^*)=\{\frac{\phi(r^m)}2, \frac{(\sqrt r-1){r^{m-1}}}{2},\frac{(-\sqrt r-1){r^{m-1}}}{2},0\}=\{0,1,6,14\}$, which has size four. By  Proposition \ref{eqn:prop4.7}, the partition $\{D^*_{f,i}\}_{i\in f(\mathbb{F}_q^*)}$ of $\mathbb{F}^*_q$ induces a 4-class association scheme.
 }
\end{exa}

\begin{prop} \label{eqn:prop4.9}
Assume that   $m,n$  are  two positive integers, $p_1\equiv p_2\equiv3\pmod4$ are two distinct primes with $\gcd(p_1(p_1-1),p_2(p_2-1))=2$, and a prime number $p\ge7$ is a common primitive root modulo $p_1^m$ and $p_2^n$. Let $N=p_1^mp_2^n$, $q=p^{\frac{\phi(N)}2}$,  and $f(x)=\Tr_{q/p}(x^{\frac{q-1}{N}})$.
Then the set $\{D^*_{f,i}\}_{i\in f(\mathbb{F}_q^*)}$ induces a $d$-class association scheme
if and only if the size of
$$
f(\mathbb{F}^*_q)=\bigg\{\frac{\phi(N)}2,
{-\frac{\phi(N)}{2(p_1-1)}},{-\frac{\phi(N)}{2(p_2-1)}}, \frac{(1-\sqrt {p_1p_2})N}{2p_1p_2}, \frac{(1+\sqrt {p_1p_2})N}{2p_1p_2},0\bigg\}$$ is $d$, where $\sqrt {p_1p_2}$ belongs to $\Bbb F_p$ such that $\sqrt {p_1p_2}\equiv p_1p_2\pmod p$.


\end{prop}

\begin{proof} Let $\Bbb F_q^*=\langle \gamma \rangle$. Then $\Bbb F_p^*=\langle \gamma ^{\frac{q-1}{p-1}}\rangle$. By the assumption, we have that $\gcd(p_1p_2,p-1)=1$ and $p_1^mp_2^n \mid q-1$; thus,  $N=p_1^mp_2^n\mid \frac{q-1}{p-1}$. Hence,  $f(ax)=f(a)$ for any $a\in \Bbb F_p^*$; therefore,
we can use Theorem \ref{eqn:thm3.4}. By the proof of \cite[Theorem 26]{WYL},  $$f(\Bbb F_q^*)=\bigg\{\frac{\phi(N)}2,
{-\frac{\phi(N)}{2(p_1-1)}},{-\frac{\phi(N)}{2(p_2-1)}}, \frac{(1-\sqrt {p_1p_2})N}{2p_1p_2}, \frac{(1+\sqrt {p_1p_2})N}{2p_1p_2},0\bigg\},$$ where the frequency of $0$ as the image of $f$ is $q-1-\frac{q-1}{p_1^{m-1}p_2^{n-1}}$. By \cite[Theorem 30]{WYL}, the set $\{W_f(\beta):\beta\in \Bbb F_q^*\}$ is given by
\begin{eqnarray*}\bigg\{\sqrt{q}\zeta_p^{\frac{\phi(N)}2}+C,
\sqrt{q}\zeta_p^{{-\frac{\phi(N)}{2(p_1-1)}}}+C,
\sqrt{q}\zeta_p^{{-\frac{\phi(N)}{2(p_2-1)}}}+C,\\
\sqrt{q}\zeta_p^{\frac{(1-\sqrt {p_1p_2})N}{2p_1p_2}}+C,
\sqrt{q}\zeta_p^{\frac{(1+\sqrt {p_1p_2})N}{2p_1p_2}}+C,
\sqrt{q}+C\bigg\}, \end{eqnarray*} where \begin{eqnarray*}C&=&N-1+\zeta_p^{\frac{\phi(N)}{2}}+\phi(p_1)(\zeta_p^{-\frac{\phi(N)}{2(p_1-1)}}-1)+\phi(p_2)(\zeta_p^{-\frac{\phi(N)}{2(p_2-1)}}-1)\\&+&(\zeta_p^{\frac{(1-\sqrt {p_1p_2})N}{2p_1p_2}}+\zeta_p^{\frac{(1+\sqrt {p_1p_2})N}{2p_1p_2}}-2)\frac{\phi(p_1p_2)}{2}.\end{eqnarray*}
The result thus follows from Corollary 3.5.
\end{proof}

\begin{exa} {\rm
(1) Let $p=2$, $p_1=3$, $p_2=11$, and $m=n=1$. Then $N=33$, $q=2^{10}$, and $f(x)=\Tr_{2^{10}/2}(x^{\frac{2^{10}-1}{33}})$. Then $f(\Bbb F_q^* )=\{0,1\}$. By  Proposition \ref{eqn:prop4.9}, the partition $\{D^*_{f,i}\}_{i\in f(\mathbb{F}_q^*)}$ of $\mathbb{F}^*_q$ induces a 2-class association scheme.

(2) Let $p=101$, $p_1=3, p_2=11$, and $m=n=1$. Then    $N=33$,  $q=101^{10}$,  and $f(x)=\Tr_{101^{10}/101}(x^{\frac{101^{10}-1}{33}})$.  Then $$f(\Bbb F_q^* )=\bigg\{\frac{\phi(N)}2,
{-\frac{\phi(N)}{2(p_1-1)}},{-\frac{\phi(N)}{2(p_2-1)}}, \frac{(1-\sqrt {p_1p_2})N}{2p_1p_2}, \frac{(1+\sqrt {p_1p_2})N}{2p_1p_2}\bigg\}$$ is $\{10,15, 87, 96, 100\}$. By  Proposition \ref{eqn:prop4.9}, the partition $\{D^*_{f,i}\}_{i\in f(\mathbb{F}_q^*)}$ of $\mathbb{F}^*_q$ induces a 5-class association scheme.

 (3) Let $p=101$, $p_1=3,p_2=11$, and $m=n=2$. Then   $N=33^2$,  $q=101^{330}$,  and $f(x)=\Tr_{101^{330}/101}(x^{\frac{101^{330}-1}{33^2}})$.  Then $$f(\Bbb F_q^* )=\bigg\{\frac{\phi(N)}2,
{-\frac{\phi(N)}{2(p_1-1)}},{-\frac{\phi(N)}{2(p_2-1)}}, \frac{(1-\sqrt {p_1p_2})N}{2p_1p_2}, \frac{(1+\sqrt {p_1p_2})N}{2p_1p_2},0\bigg\}$$ is $\{0,27, 35, 37, 45, 91\}$. By  Proposition \ref{eqn:prop4.9}, the partition $\{D^*_{f,i}\}_{i\in f(\mathbb{F}_q^*)}$ of $\mathbb{F}^*_q$ induces a 6-class association scheme.
}
\end{exa}


\subsection{Construction of two-weight linear codes}$~$\label{subsection4.3}


An $[n, k, d]$ linear code $\mathcal{C}$ over the finite field $\Bbb F_p$ 
is a $k$-dimensional subspace of $\Bbb F_p^n$ with minimum  Hamming  distance $d$.
  Let $A_i$ be the number of codewords in $\mathcal C$
 with Hamming weight $i$. The weight enumerator of
 $\mathcal C$ is defined by
 $1+A_1z+\cdots+A_nz^n.$
The sequence $(1, A_1, A_2, \ldots, A_n)$ is called the \emph{weight distribution} of
 $\mathcal C$. We say that a code $\mathcal{C}$ has \emph{$t$-weight} if the number of nonzero $A_{i}$'s 
 in the sequence $(A_1, A_2, \ldots, A_n)$ is equal to $t$.

Let $D$ be a subset of $\mathbb{F}^*_q$.
A linear code $\mathcal{C}_D$ of length $|D|$ over the finite field $\Bbb F_p$ is defined by
\begin{equation}\label{eqn:4.4}
 \mathcal{C}_D = \{c_D(\beta)= (\operatorname{Tr}(\beta \alpha))	 _{\alpha\in D} \; : \; \beta\in\mathbb{F}_q \}.
\end{equation}
The code $\mathcal{C}_D$ is called a \emph{trace code} over $\mathbb{F}_p$, and the set $D$ is called the \emph{defining set} of $\mathcal{C}_D$. This generic construction was first introduced by Ding {\em et al.} \cite{D1, DN}.  Recent results on this generic construction of linear codes from cryptographic functions are due to the references \cite{CK, CCZ, CDY, D2, DW, M1, MOS, MS,TLQZH, YCD, ZLFH}.

The dimension of the linear code $\mathcal{C}_D$ is determined by the following lemma.

\begin{lem}\cite[Theorem 6]{DN} Let $D$ be a subset of $\mathbb{F}^*_q$. The dimension of the linear code $\mathcal{C}_D$ is equal to the dimension
of $\left\langle D\right\rangle$, where $\left\langle D\right\rangle$ is the subgroup of $\mathbb{F}^*_q$
generated by $D$.
\end{lem} For each $\beta\in\mathbb{F}_q^*$, the weight of $c_D(\beta)\in\mathcal{C}_D$ is given by
\begin{multline}\label{5}
\operatorname{wt}(c_D(\beta))
= |D| - \frac{1}{p}\sum_{\alpha\in D}\sum_{a\in\mathbb{F}_p}\zeta_p^{a\Tr(\alpha\beta)}
=(1-\frac{1}{p})|D|-\frac{1}{p}\sum_{a\in\mathbb{F}^*_p}(\chi_{\beta}(D))^{\sigma_a}.
\end{multline}

\begin{lem} Let $f$ be a $p$-ary function satisfying that $f(ax)=f(x)$ for all $a\in \Bbb F_p^*$.
Then for $i\in\mathbb{F}_p$,
\[
\operatorname{wt}(c_{D^*_{f,i}}(\beta))
=\frac{p-1}{p}|D^*_{f,i}|-\frac{p-1}{p^2} \sum_{y\in \Bbb F_p^*}(\zeta^{-i}_pW_f(\beta))^{\sigma_y}.
\]
\end{lem}

\begin{proof}
From Corollary \ref{eqn:cor3.2} and the assumption on $f$, it follows that
\begin{multline}\label{eqq1}
\chi_{\beta}(D^*_{f,i})
=\frac{1}{p^2}\sum_{j\in\mathbb{F}_p}\zeta^j_p\sum_{y\in\mathbb{F}^*_p}\left(\zeta_p^{-i}\sum_{z\in\mathbb{F}^*_p}\zeta_p^{jz}W_f(z \beta)\right)^{\sigma_{y}}
=\frac{1}{p^2}\sum_{j\in\mathbb{F}_p}\zeta^j_p\sum_{y\in\mathbb{F}^*_p}\left(\zeta_p^{-i}\sum_{z\in\mathbb{F}^*_p}\zeta_p^{jz}W_f( \beta)\right)^{\sigma_{y}}\\
=\frac{1}{p^2}\sum_{y\in\mathbb{F}^*_p}\left(\zeta_p^{-i}\sum_{z\in\mathbb{F}^*_p}W_f( \beta)\right)^{\sigma_{y}}\sum_{j\in\mathbb{F}_p}\zeta^j_p(-1+p\delta_{0,j})
=\frac{1}{p}\sum_{y\in\mathbb{F}^*_p}\left(\zeta_p^{-i}W_f( \beta)\right)^{\sigma_{y}}.
\end{multline}
By Lemma \ref{eqn:lem4.1}, (\ref{5}), we get
\begin{equation}\label{eqq2}
\operatorname{wt}(c_{D^*_{f,i}}(\beta))
=(1-\frac{1}{p})|D^*_{f,i}|-\frac{1}{p}\sum_{a\in\mathbb{F}^*_p}(\chi_{\beta}(D^*_{f,i}))^{\sigma_a}
=(1-\frac{1}{p})|D^*_{f,i}|-\frac{p-1}{p}\chi_{\beta}(D^*_{f,i}).
\end{equation}
The proof completes by combing Eq.~\eqref{eqq1} with Eq.~\eqref{eqq2}.
\end{proof}

\begin{prop} Under the same assumption as Proposition \ref{eqn:prop4.5}, 
we obtain the following four classes of two-weight linear codes defined in Equation (\ref{eqn:4.4}).

(1) $\mathcal{C}_{D_{f,-r^{m-1}}}$ is a two-weight linear code with parameters $$\bigg[\frac{q-1}{r^{m-1}}-\frac{q-1}{r^{m}}, \phi(r^m), \frac{(p-1)(q+\sqrt{q})(r-1)}{pr^m}-\frac{(p-1)\sqrt{q}}{p}\bigg]$$ and its weight distribution is given in Table 1.

(2) If $r\not\equiv 1\pmod p$, then
$\mathcal{C}_{D_{f,\phi(r^m)}}$ is a two-weight linear code with parameters $$\bigg[\frac{q-1}{r^m}, \phi(r^m), \frac{(p-1)(q+\sqrt{q})}{pr^m}-\frac{(p-1)\sqrt{q}}{p}\bigg]$$ and its weight distribution is given in Table 2.

(3) If $r\equiv 1\pmod p$, then   $\mathcal{C}_{D^*_{f,0}}$ is a two-weight linear codes with parameters $$\bigg[q-1-\frac{q-1}{r^{m-1}}+\frac{q-1}{r^m}, \phi(r^m), \frac{(p-1)(q-1)}{pr^m}+\frac{p-1}{p}\cdot(q-1-\frac{q+\sqrt{q}}{r^{m-1}})\bigg]$$ and its weight distribution is given in Table 3. If $r\not\equiv 1\pmod p$, then $\mathcal{C}_{D^*_{f,0}}$ is a two-weight linear codes with parameters $$\bigg[q-1-\frac{q-1}{r^{m-1}}, \phi(r^m), \frac{p-1}{p}\cdot(q-1-\frac{q+\sqrt{q}}{r^{m-1}})\bigg]$$ and its weight distribution is given in Table 4.

\end{prop}
\[ \begin{tabular} {c} Table $1$. The weight distribution of the code $\mathcal{C}_{D_{f,-r^{m-1}}}$ in Proposition 4.14 (1)  \\ 
\begin{tabular}{c|c}
\hline
Weight&Frequency\\
\hline
$0$&$1$\\

$\frac{(p-1)(q+\sqrt{q})(r-1)}{pr^m}$&$q-1-\frac{q-1}{r^{m-1}}+\frac{q-1}{r^m}$ \\ 

$\frac{(p-1)(q+\sqrt{q})(r-1)}{pr^m}-\frac{(p-1)\sqrt{q}}{p}$&$\frac{q-1}{r^{m-1}}-\frac{q-1}{r^{m}}$\\ 
\hline
\end{tabular}
\end{tabular}
\]

\[ \begin{tabular} {c} Table $2$. The weight distribution of the code $\mathcal{C}_{D_{f,\phi(r^m)}}$ in Proposition 4.14 (2)  \\ 
\begin{tabular}{c|c}
\hline
Weight&Frequency\\
\hline
$0$&$1$\\ 

$ \frac{(p-1)(q+\sqrt{q})}{pr^m}$&$q-1-\frac{q-1}{r^m}$ \\ 

$\frac{(p-1)(q+\sqrt{q})}{pr^m}-\frac{(p-1)\sqrt{q}}{p}$&$\frac{q-1}{r^m}$\\ 
\hline
\end{tabular}
\end{tabular}
\]

\[ \begin{tabular} {c} Table $3$. The weight distribution of the code $\mathcal{C}_{D^*_{f,0}}$ in Proposition 4.14 (3)  \\ 
\begin{tabular}{c|c}
\hline
Weight&Frequency\\
\hline
$0$&$1$\\ 

$\frac{(p-1)(q-1)}{pr^m}+\frac{p-1}{p}(q-1-\frac{q+\sqrt{q}}{r^{m-1}})$&$q-1-\frac{q-1}{r^{m-1}}+\frac{q-1}{r^{m}}$ \\ 

$\frac{(p-1)(q-1)}{pr^m}+\frac{p-1}{p}(q+\sqrt{q}-1-\frac{q+\sqrt{q}}{r^{m-1}})$&$\frac{q-1}{r^{m-1}}-\frac{q-1}{r^{m}}$\\ 
\hline
\end{tabular}
\end{tabular}
\]

\[ \begin{tabular} {c} Table $4$. The weight distribution of the code $\mathcal{C}_{D^*_{f,0}}$ in Proposition 4.14 (3)   \\ 
\begin{tabular}{c|c}
\hline
Weight&Frequency\\
\hline
$0$&$1$\\

$\frac{p-1}{p}(q-1-\frac{q+\sqrt{q}}{r^{m-1}})$&$q-1-\frac{q-1}{r^{m-1}}$ \\ 

$\frac{p-1}{p}(q+\sqrt{q}-1-\frac{q+\sqrt{q}}{r^{m-1}})$&$\frac{q-1}{r^{m-1}}$\\  
\hline
\end{tabular}
\end{tabular}
\]

\

\begin{proof} By \cite[Theorem 10]{WYL}, we have $|D^*_{f,0}|=|D_{f,0} \backslash \{0\}|=q-1-\frac{q-1}{r^{m-1}}$, $|D_{f,\phi(r^m)}|=\frac{q-1}{r^{m}}$, and $|D_{f,-r^{m-1}}|=\frac{q-1}{r^{m-1}}-\frac{q-1}{r^{m}}$.
By the proof of Proposition 4.6, we have $f(\Bbb F_q^*)=\{\phi(r^m), -r^{m-1},0\}$ and $$\{W_f(\beta):\beta\in \Bbb F_q^*\}=\{\sqrt{q}\zeta_p^{\phi(r^m)}+A, \sqrt{q}\zeta_p^{-r^{m-1}}+A, \sqrt{q}+A\},$$ where $A=1-\frac{\sqrt{q}+1}{r^m}(r^m-1+\zeta_p^{\phi(r^m)}+\phi(r)(\zeta_p^{-r^{m-1}}-1))$.  

It is thus sufficient to prove the case where the size of
$f(\Bbb F_q^*)=\{\phi(r^m), -r^{m-1},0\} \pmod p $ is three.
We consider the following three cases for the proof.

Case 1. $i=0$. In this case, we consider the following two subcases:

(a) If $\beta\in D^*_{f,0}$, then $W_f(\beta)=\sqrt q+A$. By noting that $\sigma_y(a)=a$ for $a\in \mathbb Q$ and $\sigma_y(\zeta_p^i)=\zeta_p^{iy}$ for $i\in \Bbb F_p^*$, we have $\sum_{y\in \Bbb F_p^*}\sigma_y(\zeta_p^i) =-1$ and
\begin{eqnarray*}&&\sum_{y\in \Bbb F_p^*}(W_f(\beta))^{\sigma_y}=\sum_{y\in \Bbb F_p^*}(\sqrt q+A)^{\sigma_y}\\
&=&(p-1)(\sqrt q+1-\frac{\sqrt{q}+1}{r^m}(r^m-1)+\frac{\sqrt{q}+1}{r^m}(r-1))+\frac{\sqrt{q}+1}{r^m}(1+r-1)  \\
&=&\frac{p(\sqrt {q}+1)}{r^{m-1}}.
\end{eqnarray*}
As a result, we get
\begin{eqnarray*}\operatorname{wt}(c_{D^*_{f,i}}(\beta))&=&\frac{p-1}{p}(q-1-\frac{q-1}{r^{m-1}})-\frac{p-1}{p^2}\cdot\frac{p(\sqrt {q}+1)}{r^{m-1}}=\frac{p-1}{p} (q-1-\frac{q+\sqrt{q}}{r^{m-1}}).\end{eqnarray*}

(b) If $\beta\in D_{f,\phi(r^m)} $ or $D_{f,-r^{m-1}}$, then $W_f(\beta)=\sqrt{q}\zeta_p^{\phi(r^m)}+A$ or $\sqrt{q}\zeta_p^{-r^{m-1}}+A$. Hence, we have that \begin{eqnarray*}&&\sum_{y\in \Bbb F_p^*}(W_f(\beta))^{\sigma_y}=\sum_{y\in \Bbb F_p^*}(\sqrt q \zeta_p^{\phi(r^m)}+A)^{\sigma_y}\\
&=&-\sqrt q+(p-1)(1-\frac{\sqrt{q}+1}{r^m}(r^m-1)+\frac{\sqrt{q}+1}{r^m}(r-1))+\frac{\sqrt{q}+1}{r^m}(1+r-1)  \\
&=&\frac{p(\sqrt {q}+1)}{r^{m-1}}-p\sqrt{q}
\end{eqnarray*}
or \begin{eqnarray*}&&\sum_{y\in \Bbb F_p^*}(W_f(\beta))^{\sigma_y}=\sum_{y\in \Bbb F_p^*}(\sqrt q \zeta_p^{-r^{m-1}}+A)^{\sigma_y}\\
&=&-\sqrt q+(p-1)(1-\frac{\sqrt{q}+1}{r^m}(r^m-1)+\frac{\sqrt{q}+1}{r^m}(r-1))+\frac{\sqrt{q}+1}{r^m}(1+r-1)  \\
&=&\frac{p(\sqrt {q}+1)}{r^{m-1}}-p\sqrt{q}.
\end{eqnarray*}
That is to say, we have
$\sum_{y\in \Bbb F_p^*}(W_f(\beta))^{\sigma_y}=\frac{p(\sqrt {q}+1)}{r^{m-1}}-p\sqrt{q}$. Then \begin{eqnarray*}\operatorname{wt}(c_{D^*_{f,i}}(\beta))&=&\frac{p-1}{p}(q-1-\frac{q-1}{r^{m-1}})-\frac{p-1}{p^2}(\frac{p(\sqrt {q}+1)}{r^{m-1}}-p\sqrt{q})\\
&=&\frac{p-1}{p}(q+\sqrt{q}-1-\frac{q+\sqrt{q}}{r^{m-1}}).\end{eqnarray*}

Case 2. $i=\phi(r^m)$. In this case, we have the following two subcases:

(a) If $\beta\in D^*_{f,0}$ or $\beta\in D_{f,-r^{m-1}}$, then $\sum_{y\in \Bbb F_p^*}(\zeta_p^{-\phi(r^m)}W_f(\beta))^{\sigma_y}=-\frac{p(\sqrt{q}+1)}{r^m}$; thus, we get
\begin{eqnarray*}\operatorname{wt}(c_{D_{f,i}}(\beta))&=&\frac{p-1}{p}\frac{q-1}{r^{m}}+\frac{p-1}{p}\frac{\sqrt{q}+1}{r^m}=\frac{(p-1)(q+\sqrt{q})}{pr^m}.
\end{eqnarray*}

(b) If $\beta\in D_{f,\phi(r^m)}$, then $\sum_{y\in \Bbb F_p^*}(\zeta_p^{-\phi(r^m)}W_f(\beta))^{\sigma_y}=p\sqrt{q}-\frac{p(\sqrt{q}+1)}{r^m}$; so we have
\begin{eqnarray*}\operatorname{wt}(c_{D_{f,i}}(\beta))&=&\frac{p-1}{p}\frac{q-1}{r^{m}}-\frac{p-1}{p^2}(p\sqrt{q}-p\frac{\sqrt{q}+1}{r^m})\\
&=&\frac{(p-1)(q+\sqrt{q})}{pr^m}-\frac{(p-1)\sqrt{q}}{p}.
\end{eqnarray*}

Case 3.  $i=-r^{m-1}$. In this case, there are two subcases to consider as follows:

(a) If $\beta\in D^*_{f,0}$ or $\beta\in D_{f,\phi(r^m)}$, then $\sum_{y\in \Bbb F_p^*}(\zeta_p^{r^{m-1}}W_f(\beta))^{\sigma_y}=-\frac{p(\sqrt{q}+1)(r-1)}{r^m}$; therefore, we have
\begin{eqnarray*}\operatorname{wt}(c_{D_{f,i}}(\beta))&=&\frac{p-1}{p}(\frac{q-1}{r^{m-1}}-\frac{q-1}{r^{m}})+\frac{p-1}{p}\frac{(\sqrt{q}+1)(r-1)}{r^m}\\
&=&\frac{(p-1)(q+\sqrt{q})(r-1)}{pr^m}.
\end{eqnarray*}


(b) If $\beta\in D_{f,-r^{m-1}}$, then $\sum_{y\in \Bbb F_p^*}(\zeta_p^{r^{m-1}}W_f(\beta))^{\sigma_y}=p\sqrt{q}-\frac{p(\sqrt{q}+1)(r-1)}{r^m}$; thus, we get
\begin{eqnarray*}\operatorname{wt}(c_{D_{f,i}}(\beta))&=&\frac{p-1}{p}(\frac{q-1}{r^{m-1}}-\frac{q-1}{r^{m}})-\frac{p-1}{p^2}(p\sqrt{q}-\frac{p(\sqrt{q}+1)(r-1)}{r^m})\\
&=&\frac{(p-1)(q+\sqrt{q})(r-1)}{pr^m}-\frac{(p-1)\sqrt{q}}{p}.
\end{eqnarray*}

This completes the proof.
\end{proof}






\end{document}